\documentclass[adraft,copyright,creativecommons]{eptcs}
\providecommand{\event}{EXPRESS/SOS 2014} 
\usepackage{breakurl}             

\usepackage[all]{xy}
\usepackage{amsmath}
\usepackage{amssymb}
\usepackage{eurosym}
\usepackage{wasysym}

\newtheorem{theorem}{Theorem}

\newtheorem{definition}[theorem]{Definition}
\newtheorem{example}[theorem]{Example}

\newtheorem{lemma}[theorem]{Lemma}

\newenvironment{proof}[1][Proof]{\noindent\textbf{#1.} }{{\hfill $\Box$ \\}}

\DeclareSymbolFont{cmlargesymbols}{OMX}{cmex}{m}{n}
\DeclareMathSymbol{\mycoprod}{\mathop}{cmlargesymbols}{"60}

\newcommand{\pto}{\rightharpoonup}
\newcommand{\To}{\Longrightarrow}
\newcommand{\Bid}{\tilde{B}}
\newcommand{\Bb}{B^{\infty}}
\newcommand{\Sid}{\bar{\S}}
\newcommand{\Ss}{\Sigma^*}
\renewcommand{\S}{\Sigma}
\newcommand{\id}{\textrm{id}}
\newcommand{\Id}{\textrm{Id}}
\newcommand{\goes}[1]{\stackrel{#1}{\longrightarrow}}

\newcommand{\labA}{\$}
\newcommand{\labB}{\text{\euro}}
\newcommand{\Pf}{{\cal P}_{\omega}}

\newcommand{\abstractgoes}[2]{%
\setbox0=\hbox{\ ${\scriptstyle#2}$\ }
\ifdim\wd0<12pt\wd0=12pt\fi
\mathrel{\stackrel{#2}{\rule[2.2pt]{\wd0}{0.6pt}}\mkern-6mu{#1}}}

\newcommand{\qmgoes}[1]{\abstractgoes{\RHD}{#1}}
\newcommand{\mqmgoes}[1]{\abstractgoes{\rhd}{#1}}

\newcommand{\sosngoes}[1]{%
\setbox0=\hbox{\ ${\scriptstyle#1}$\ }
\ifdim\wd0<12pt\wd0=12pt\fi
\mathrel{\stackrel{#1}{\rule[2.2pt]{0.7\wd0}{0.6pt}}\mkern-7mu \mathord{/}%
\mkern-7mu\rule[2.2pt]{0.3\wd0}{0.6pt}\mkern-0mu\triangleright}}

\title{Distributive Laws \\ and Decidable Properties of SOS Specifications\footnote{This work was supported by the Polish National Science Centre (NCN) grant 2012/07/E/ST6/03026.}}
\author{Bartek Klin
\institute{University of Warsaw}
\email{klin@mimuw.edu.pl}
\and
Beata Nachy\l{}a
\institute{Institute of Computer Science, 
Polish Academy of Sciences}
\email{beatanachyla@gmail.com}
}
\def\titlerunning{Distributive Laws and Decidable Properties of SOS}
\def\authorrunning{B. Klin \& B. Nachy\l{}a}
\begin{document}
\maketitle

\begin{abstract}
Some formats of well-behaved operational specifications, correspond to natural transformations of certain types (for example, GSOS and coGSOS laws). These transformations have a common generalization: distributive laws of monads over comonads. We prove that this elegant theoretical generalization has limited practical benefits: it does not translate to any concrete rule format that would be complete for specifications that contain both GSOS and coGSOS rules. This is shown for the case of labeled transition systems and deterministic stream systems.

\end{abstract}

\vspace{-3ex}
\section{Introduction}

Distributive laws (see~\cite{turiplotkin,tcs11} for more information) are an abstract approach to several kinds of well-behaved operational specifications. For example, for a fixed set $A$ of labels, a family of inference rules
\[
	\dfrac{x\goes{a}x'\quad y\goes{a}y'}{x\otimes y\goes{a}x'\otimes{y'}} \qquad (\mbox{for }a\in A)
\]
that define synchronous composition over labeled transition systems (LTSs), can be presented as a natural transformation $\lambda:\S B\To B\S$ (a distributive law of $\S$ over $B$), where $\Sigma X=X\times X$ and $BX = \Pf(A\times X)$ are functors on the category $\textbf{Set}$ of sets and functions. Similarly, a family of rules
\[
	\dfrac{x\goes{a}x' \quad y\goes{b}y'}{x\rtimes y \goes{a}y'\rtimes x'} \qquad (\mbox{for }a,b\in A)
\]
that define an alternating composition operator $\rtimes$ on infinite streams of labels,
can be understood as a transformation $\lambda:\S B\To B\S$ where $\Sigma X=X\times X$ again, and $BX=A\times X$.

Typically $\S$ is a polynomial functor arising from an algebraic signature. Specifications that give rise to distributive laws of $\S$ over $B$ enjoy several desirable properties: they induce a $B$-coalgebra (e.g.~an LTS) on the carrier of the initial $\Sigma$-algebra (the algebra of $\Sigma$-terms) so that bisimilarity is a congruence, and they provide an interpretation of the signature on the final $B$-coalgebra (provided that it exists).

These desirable properties extend to other, more expressive types of laws, including:
\begin{itemize}
\item[(a)] {\em GSOS laws} $\rho:\S(B\times\Id)\To B\Ss$, where $\Ss$ is the free monad over $\S$ (see Section~\ref{sec:algebras}),
\item[(b)] {\em coGSOS laws} $\rho:\S\Bb\To B(\Id+\S)$, where $\Bb$ is the cofree comonad over $B$ (see Section~\ref{sec:coalgebras}),
\item[(c)] {\em distributive laws of monads over comonads}, i.e., natural transformations $\lambda:\Ss\Bb\To\Bb\Ss$ subject to a few axioms.  (In this paper we only consider distributive laws of free monads over cofree comonads, see Section~\ref{sec:dist-laws}.)
\end{itemize}
GSOS and coGSOS laws are incomparable, i.e., there are specifications that conform to one type but not the other, and distributive laws of monads over comonads are a common generalization of both. From now on, for brevity, we shall call them simply {\em distributive laws}. 

For standard examples of $B$, GSOS and coGSOS laws correspond to {\em rule formats}, i.e., syntactic restrictions on the form of inference rules that are allowed in a specification for it to define a corresponding type of law. For $BX=\Pf(A\times X)$, where $\Pf$ is the finite powerset functor, it was observed in~\cite{turiplotkin} that GSOS laws correspond to previously known GSOS~\cite{gsos} specifications (hence the name of the law type), that allow rules such as:
\[
	\dfrac{x_1\goes{a_{1,1}}y_{1,1} \quad x_1\goes{a_{1,2}}y_{1,2} \quad\cdots\quad x_i\goes{a_{i,j}}y_{i,j} \quad \cdots\quad x_i\not\goes{b_{i,j}}\quad\cdots}{{\tt f}(x_1,\ldots, x_k)\goes{b} {\tt t}}
\]
where variables $x_i$ can be tested for the presence and/or absence of transitions labeled with different labels, and the resulting transition can go to an arbitrary term ${\tt t}$ built over the variables $x_i$ and $y_{i,j}$. On the other hand, coGSOS laws for the same functor $B$ are induced by {\em safe ntree}~\cite{ntree,turiplotkin} specifications, where additionally {\em lookahead} is allowed, i.e., variables that are targets of premise transitions can be further tested for other transitions as in the rule:
\[
	\dfrac{x\goes{a} y \goes{b}z}{{\tt f}(x)\goes{c}{\tt g}(z)}
\]
On the other hand, coGSOS is restricted in that the target term ${\tt t}$ in the conclusion must be either a variable or a flat term built of a single operation symbol and variables.

Both GSOS and coGSOS laws, are generalized by distributive laws. In fact, desired properties of systems induced by GSOS and coGSOS laws were proved in~\cite{turiplotkin} by showing first that these laws induce distributive laws, and then proving those properties for the latter, more general laws. This is tantalizing, as it suggests that for standard functors $B$ one could find new, more expressive syntactic rule formats that would correspond to distributive laws and hence guarantee good properties of specifications. The problem of finding such a format was left open in~\cite{turiplotkin} and mentioned as still open in later works~\cite{bartelsthesis,tcs11}.

The purpose of this paper is to suggest a negative answer to that problem. Specifically, we claim that there is no rule format
that would adequately recognize those specifications that induce distributive laws of monads over comonads, within a class of specifications that extends both GSOS and coGSOS.

This claim is rather vague, and we must make it precise before we attempt to prove it. First of all, there is no hope to prove it for all monads and comonads; clearly, for some trivial monads and comonads all distributive laws are easily enumerated, and even for some nontrivial comonads a complete description of distributive laws is known~\cite{kickformat}. 
Therefore in this paper we shall consider laws $\lambda:\Ss\Bb\to\Bb\Ss$ for $\Ss$ the free monad over a polynomial functor $\S$, and $\Bb$ the cofree comonad over
 $BX = A\times X$, pertaining to stream systems, or
$BX = \Pf(A\times X)$, pertaining to labeled transition systems.
Hopefully it shall be clear how our arguments for the lack of expressive formats for these two behaviour functors, might extend to other standard functors used to model transition systems coalgebraically.

To make our claim precise, the first question we need to answer is: {\em what is a format?} In positive results about GSOS and coGSOS laws mentioned above, the answer was easy: one simply formulated some ``syntactic forms'' of rules and provided ways of defining laws from sets of rules that conformed to them. Now that we want a negative result, we need to quantify over all ``syntactic forms'', so we need to understand what a syntactic form is in general. We opt for a general and permissive answer:
{\em a format is a decidable property of specifications.} Indeed, no matter what a ``format'' may be, it should be effectively checkable whether a specification conforms to it. 

This leads to another question: {\em what is a specification?} Some definitions of this term would immediately invalidate our claim; for example, if we say that ``a specification is either a GSOS specification or a coGSOS specification'', then every specification induces a distributive law as described already in~\cite{turiplotkin} and the problem is trivially decidable. However, we are interested in more permissive notions of specification that would allow a more substantial combination of GSOS and coGSOS features.
We therefore focus on {\em mixed-GSOS specifications}, where every rule is either a GSOS or a coGSOS rule.

Note that there are other interesting notions of specification where the claim becomes false. For example, as proved in~\cite{statontyft}, in the context of LTSs one may consider so-called (positive) {\em tyft/tyxt}~\cite{tyft} specifications, and guarantee the existence of a distributive law for every specification. However, tyft/tyxt specifications extend neither full GSOS nor coGSOS, so this does not match the abstract observation that distributive laws generalize both GSOS laws and coGSOS laws.

There is still one vague point in our claim: {\em what does it mean for a specification to induce a distributive law?}
In positive results about GSOS and coGSOS specifications~\cite{tcs11,turiplotkin}, one simply provides particular ways of inducing distributive laws from specifications that look so natural that everybody is convinced. Here, to show undecidability, we shall need to prove that some instances do {\em not} induce distributive laws, so we need to quantify over all possible ``ways of inducing laws'', a  vague notion itself.

We approach this problem by observing that every mixed-GSOS specification induces, in a very natural way, a natural transformation $\rho:\S \Bb\To B\Ss$ which we call a {\em biGSOS law}. Then we define (Definition~\ref{def:extn}) what it means for a distributive law  $\lambda$ to {\em extend} a biGSOS law $\rho$; essentially, $\lambda$ must restrict to $\rho$ when composed with obvious inclusions and projections. Our claim then becomes:

\noindent
{\bf Claim.} 
{\em It is undecidable whether a given mixed-GSOS specification extends to a unique distributive law.}

One may worry whether our insistence on a unique extension is not overly restrictive. Indeed, perhaps sometimes a specification may extend to several distributive laws, but one of these laws is somehow better than the other ones, for example (in the LTS setting) the least one, or canonical in some other way? However, as will be evident from our proofs, this is not a problem: all our instances of specifications will either extend to one distributive law or to none at all, therefore no matter what notion of ``canonical extension'' one may come up with, the problem remains undecidable.

We prove undecidability by reduction from the halting problem of a variant of queue machines defined in Section~\ref{sec:queue}. Then, in Section~\ref{sec:qm2ss}, we prove the Claim for the case of stream systems (Theorem~\ref{thm:ssundec}), and in Section~\ref{sec:qm2lts} we explain how the proof is adapted to the case of LTSs (Theorem~\ref{thm:ltsundec}).


\section{Preliminaries}

The reader should be familiar with notions of category theory such as functors and natural transformations, see e.g.~\cite{maclane}. 
All functors we consider are endofunctors on the category of sets and functions.

\subsection{Algebras and monads}\label{sec:algebras}

An {\em algebra} for a functor $\S$ is a set $X$ (the {\em carrier}) together with a function $g:\S X\to X$ (the {\em structure}). An algebra morphism from $g:\S X\to X$ to $h:\S Y \to Y$ is a function $f:X\to Y$ such that $f\circ g = h\circ \S f$. Algebras for $\S$ and their morphisms form a category. Of particular interest in this category are initial objects, i.e., initial $\S$-algebras.

Assume that, for any set $X$, an initial algebra for the functor $\S(-)+X$ exists, denote its carrier $\Ss X$ and its structure by:
\[
	\xymatrix{\S\Ss X\ar[r]^{\psi_X} & \Ss X & X\ar[l]_-{\eta_X}.}
\]
Then $\Ss$, defined on functions using initiality, becomes a functor and $\psi:\S\Ss\To\Ss$ and $\eta:\Id\To\Ss$ are natural transformations. Moreover, $\Ss$ is a {\em monad}, i.e., it is equipped with a natural transformation $\mu:\Ss\Ss\To\Ss$ such that the following diagrams commute:
\begin{equation}\label{eq:monadlaws}
\vcenter{
	\xymatrix{\Ss\ar@{=>}[r]^-{\Ss\eta}\ar@{=}[rd] & \Ss\Ss\ar@{=>}[d]^{\mu} & \Ss\ar@{=>}[l]_-{\eta\Ss}\ar@{=}[ld] & & \Ss\Ss\Ss\ar@{=>}[r]^{\Ss\mu}\ar@{=>}[d]_{\mu\Ss} & \Ss\Ss\ar@{=>}[d]^{\mu} \\ & \Ss & & & 
 \Ss\Ss\ar@{=>}[r]_{\mu} & \Ss.}
}
\end{equation}
$\Ss$ is called the {\em free monad} over $\S$. Another relevant transformation is $\iota:\S\To\Ss$ defined by $\iota = \psi \circ \S\eta$; it further satisfies the equation $\psi = \mu\circ\iota\Ss$.

\begin{example}\label{ex:syntax}\rm
Any algebraic signature $({\tt q}_i)_{i\in I}$, where each ${\tt q}_i$ is an operation symbol of arity $n_i\in\mathbb{N}$, gives rise to an endofunctor $\Sigma X = \mycoprod_{i\in I}X^{n_i}$. Then $\S$-algebras are algebras for the signature in the sense of universal algebra, and $\S$-algebra morphisms are exactly algebra homomorphisms. Moreover, $\Ss X$ is the set of terms over the signature with variables taken from $X$, $\eta$ interprets variables as terms, $\psi$ and $\mu$ glue together terms built of terms, and $\iota$ interprets terms built of single operation symbols as terms.
\end{example}

\subsection{Coalgebras and comonads}\label{sec:coalgebras}

The following development is dual to the one for algebras and monads; we include it for completeness and to introduce some basic terminology and notation. For more information about coalgebras, see~\cite{ruttentcs}.

A {\em coalgebra} for a functor $B$ is a set $X$ (the {\em carrier}) together with a function $g:X\to BX$ (the {\em structure}). A coalgebra morphism from $g:X\to BX$ to $h:Y \to BY$ is a function $f:X\to Y$ such that $h\circ f = Bf\circ g$. Coalgebras for $B$ and their morphisms form a category.

Assume that, for any set $X$, a final coalgebra for the functor $B(-)\times X$ exists, denote its carrier $\Bb X$ and its structure by:
\[
	\xymatrix{B\Bb X & \Bb X\ar[l]_-{\theta_X}\ar[r]^-{\epsilon_X} & X.}
\]
Then $\Bb$, defined on functions using finality, becomes a functor and $\theta:\Bb\To B\Bb$ and $\epsilon:\Bb\To\Id$ are natural transformations. Moreover, $\Bb$ is a {\em comonad}, i.e., it is equipped with a natural transformation $\delta:\Bb\To\Bb\Bb$ such that diagrams dual to~\eqref{eq:monadlaws} commute.
$\Bb$ is called the {\em cofree comonad} over $B$. Another relevant transformation is $\pi:\Bb\To B$ defined by $\pi = B\epsilon \circ \theta$; it further satisfies the equation $\theta = \pi\Bb\circ\delta$.

\begin{example}\label{ex:stream-coalg}\rm
Let $BX = A\times X$, for a fixed set $A$ of {\em labels}. $B$-coalgebras are {\em stream systems}, i.e., sets $X$ (of {\em states}) equipped with functions to $A$ and to $X$ again; the intuition is that a state produces a label and transforms into another state. The cofree comonad over $B$ is given by $\Bb X = (X\times A)^{\omega}$; we will depict elements of $\Bb X$ as streams of labeled transitions:
\[
	\Bb X\ni \sigma = x_0 \goes{a_0} x_1 \goes{a_1} x_2 \goes{a_2} x_3 \goes{a_3} \cdots
\]
with $x_i\in X$ and $a_i\in A$. For any $n\in\mathbb{N}$, by $\sigma^{(n)}\in\Bb X$ denote the $n$-th tail of $\sigma$, i.e., the substream of $\sigma$ that starts at $x_n$. Natural transformations explained above are then given by:
\begin{align*}
	\epsilon_X(\sigma) &= x_0 & \theta_X(\sigma) &= \big( a_0, \sigma^{(1)} ) \\
	\delta_X(\sigma) &= \big( \sigma \goes{a_0} \sigma^{(1)} \goes{a_1} \sigma^{(2)} \goes{a_2} \sigma^{(3)} \goes{} \cdots ) & 
	\pi_X(\sigma) &= (a_0,x_1)
\end{align*}
One may look at elements of $\Bb X$ as streams of labels ``colored'' with elements of $X$; elements of $\Bb\Bb X$ are then streams colored by streams, and $\delta_X(\sigma)$ is the stream that arises from $\sigma$ by coloring each node with the substream of $\sigma$ that starts in it.
\end{example}

\begin{example}\rm
Let $\Pf$ denote the finite powerset functor, and let $BX ={\cal P_\omega}(A\times X)$, for a fixed set $A$ of labels. $B$-coalgebras are {\em (finitely branching) labeled transition systems}.
The cofree comonad over $B$ is a functor $\Bb$ that maps a set $X$ to the set of finitely branching, but possibly infinitely deep trees, edge-labeled with elements of $A$ and node-colored by elements of $X$, quotiented by a version of strong bisimilarity that takes into account both edge labels and node colors.


Natural transformations listed above are defined by analogy to Example~\ref{ex:stream-coalg}. For a tree ${\tt T}\in\Bb X$:
\begin{itemize}
\item $\epsilon_X({\tt T})\in X$ is the color of the root node of ${\tt T}$,
\item $\delta_X({\tt T})\in \Bb\Bb X$ arises from ${\tt T}$ by coloring every node with the subtree rooted in it,
\item $\theta_X({\tt T})\in B\Bb X$ is the set of immediate subtrees of the root together with labels of the edges that lead to these subtrees,
\item $\pi_X({\tt T})\in BX$ is similar, but with the immediate subtrees replaced by the colors of their roots.
\end{itemize}
A little care is needed to show that components of these transformations are well-defined on bisimilarity classes of trees. For example, if ${\tt T}_1$ and ${\tt T}_2$ are related by a bisimulation, then $\delta_X({\tt T}_1)$ and $\delta_X({\tt T}_2)$ also are, as bisimilar nodes get assigned the same colors (here colors are bisimilarity classes of trees).


\end{example}

\subsection{GSOS and coGSOS laws}\label{sec:gsos}

Algebras, coalgebras, monads and comonads can be combined in distributive laws of various kinds. We only recall a few basic definitions and examples here; for a more comprehensive treatment see~\cite{tcs11}.

For any functor $B$, denote $\Bid=\Id\times B$.

\begin{definition}\rm
Given endofunctors $\S$ and $B$ such that the free monad $\Ss$ over $\S$ exists, a {\em GSOS law} is a natural transformation $\rho:\S\Bid\To B\Ss$.
\end{definition}

\begin{example}\label{ex:gsos}\rm
Consider $BX=A\times X$, and let $\S X=X\times X$ arise from a signature with a single binary function symbol ${\tt zip}$. A family of rules
\[
	\dfrac{x\goes{a}x' \quad y\goes{b}y'}{{\tt zip}(x,y) \goes{a}{\tt zip}(y,x')} \qquad (\mbox{for }a,b\in A)
\]
together defines a GSOS law by:
\[
	\rho_X({\tt zip}((x,(a,x')),(y,(b,y')))) = (a,{\tt zip}(y,x'))
\]
for any $x,x',y,y'\in X$ and $a,b\in A$. 
%
\end{example}
\medskip

Dually, for any functor $\S$, denote $\Sid=\Id+\S$.

\begin{definition}\rm
Given endofunctors $\S$ and $B$ such that the cofree comonad $\Bb$ over $B$ exists, a {\em coGSOS law} is a natural transformation $\rho:\S\Bb\To B\Sid$.
\end{definition}

\begin{example}\label{ex:co-gsos}\rm
Consider $BX=A\times X$, 
and let $\S X=X$ arise from a signature with a single unary function symbol ${\tt q}$.
The family of rules (that define a unary operation that drops every second label from a given stream):
\begin{align*}
	\dfrac{x\goes{a_1}x'\goes{a_2}x''}{{\tt q}(x)\goes{a_2}{\tt q}(x'')}\ (\mbox{for } a_1,a_2\in A)
\end{align*}
together defines a coGSOS law by:
\[
	\rho_X({\tt q}(x\goes{a_1} x' \goes{a_2} x''\goes{a_3}\cdots)) = (a_2,{\tt q}(x''))
\]
for any $x,x',x'',\hdots \in X$ and $a_1,a_2,a_3,\hdots \in A$.
\end{example}


\begin{example}\label{ex:co-gsos-neg}\rm
Now, consider the LTS behaviour functor $BX={\cal P_\omega}(A\times X)$, and let $\S X = X$ as in Example~\ref{ex:co-gsos}. The rules:
\begin{align*}
	\dfrac{x\goes{a_1}x'\goes{a_2}x''}{{\tt q}(x)\goes{a_2}{\tt q}(x'')} \qquad 
	\dfrac{x\goes{a_1}x' {\not \goes{}}}{{\tt q}(x)\goes{a_1}{\tt q}(x)}\qquad (\mbox{for } a_1,a_2\in A)
\end{align*}
define a coGSOS law $\rho:\S\Bb\To B\Sid$, where $\rho_X({\tt q}({\tt T}))$ is the set of pairs $(a,{\tt q}(x))$ such that ${\tt T}$ has (1) a two-step path from the root to a node colored by $x$ with the second step labeled by $a$, or (2) a single step, labeled with $a$, to a leaf (i.e. a node without successors), and the root of {\tt T} is colored by $x$.
%
%
\end{example}

\subsection{Distributive laws}\label{sec:dist-laws}

In~\cite{turiplotkin} it was noticed that both GSOS and coGSOS laws are generalized by distributive laws of monads over comonads; in this paper we call them simply {\em distributive laws}.

\begin{definition}\label{def:dist-law}\rm
A distributive law of a monad $(\Ss,\eta,\mu)$ over a comonad $(\Bb,\epsilon,\delta)$ is a natural transformation $\lambda:\Ss\Bb\To\Bb\Ss$ subject to the following four axioms:
\[
\xymatrix{\Bb\ar@{=>}[d]_{\eta\Bb}\ar@{=>}@/^/[rd]^{\Bb\eta}\ar@{}[rdd]|(.3){\textrm{(i)}}\ar@{}[rdd]|(.7){\textrm{(ii)}} \\
\Ss\Bb\ar@{=>}[r]^{\lambda}\ar@{=>}@/_/[rd]_{\Ss\epsilon} & \Bb\Ss\ar@{=>}[d]^{\epsilon\Ss} \\
& \Ss} \qquad\qquad
\xymatrix{\Ss\Ss\Bb\ar@{=>}[r]^{\Ss\lambda}\ar@{=>}[d]_{\mu\Bb}\ar@{}[rd]|{\textrm{(iii)}}  & \Ss\Bb\Ss\ar@{=>}[r]^{\lambda\Ss}  & \Bb\Ss\Ss\ar@{=>}[d]^{\Bb\mu} \\
\Ss\Bb\ar@{=>}[d]_{\Ss\delta}\ar@{=>}[rr]^{\lambda}  &\ar@{}[rd]|{\textrm{(iv)}}& \Bb\Ss\ar@{=>}[d]^{\delta\Ss}  \\
\Ss\Bb\Bb\ar@{=>}[r]_{\lambda\Bb}  & \Bb\Ss\Bb\ar@{=>}[r]_{\Bb\lambda}  & \Bb\Bb\Ss}
\]
\end{definition}

\begin{example}\label{ex:dist-laws}\rm
Consider $BX=A\times X$ as in Example~\ref{ex:stream-coalg}, and a functor $\S X = \mycoprod_{q\in Q}X \cong Q\times X$ arising as in Example~\ref{ex:syntax} from an algebraic signature consisting of a set of unary operation symbols $Q = \{{\tt q}_1,\ldots,{\tt q}_k\}$. Then $\Bb X = (X\times A)^{\omega}$ and $\Ss X = Q^*\times X$;  for ${\tt t}\in Q^*$, we shall write ${\tt t}(x)$ instead of $({\tt t},x)\in \Ss X$, and simply $x$ instead of $\epsilon(x)$, for the empty string $\epsilon\in Q^*$.

For a distributive law $\lambda:\Ss\Bb\To\Bb\Ss$, the naturality condition means that if
\begin{align*}
	\lambda_X({\tt t}(x_0 \goes{a_0} x_1 \goes{a_1} x_2 \goes{a_2} \cdots)) &= \tau_0\goes{b_0}\tau_1\goes{b_1}\tau_2\goes{b_2}\cdots\quad \\
	\lambda_Y({\tt t}(y_0 \goes{a_0} y_1 \goes{a_1} y_2 \goes{a_2} \cdots)) &= \gamma_0\goes{c_0}\gamma_1\goes{c_1}\gamma_2\goes{c_2}\cdots\quad
\end{align*}
then for all $i\in\mathbb{N}$ one has $b_i=c_i$, and $\gamma_i\in\Ss Y$ arises from $\tau_i\in\Ss X$ by substituting each $x_j$ by the corresponding $y_j$. Informally, the value of $\lambda_X$ on ${\tt t}(\sigma)$ essentially depends only on the term ${\tt t}$ and on the labels in the stream $\sigma$, and the colors $x_j$ in $\sigma$ are merely rearranged into terms $\tau_k$ independently from their identity or structure.
This also implies that all elements from $X$ present in $\lambda_X({\tt t}(\sigma))$ must have been present in $\sigma$.

Further, the four axioms of Definition~\ref{def:dist-law} amount to:
\begin{itemize}
\item[(i)] $\lambda_X(x_0 \goes{a_0} x_1 \goes{a_1} x_2 \goes{a_2} \cdots) \,= \, x_0 \goes{a_0} x_1 \goes{a_1} x_2 \goes{a_2} \cdots$,
\item[(ii)] if $\lambda_X({\tt t}(x_0 \goes{a_0} x_1 \goes{a_1} x_2 \goes{a_2} \cdots)) \, = \, \tau_0\goes{b_0}\tau_1\goes{b_1}\tau_2\goes{b_2}\cdots\quad$ then $\tau_0={\tt t}(x_0)$,
\item[(iii)] if $\lambda_X({\tt s}(x_0 \goes{a_0} x_1 \goes{a_1} x_2 \goes{a_2} \cdots)) \quad\!\! = \, \tau_0\goes{b_0}\tau_1\goes{b_1}\tau_2\goes{b_2}\cdots\quad$ and \\
		  \hspace*{6pt} $\lambda_{\Ss X}({\tt t}(\tau_0 \goes{b_0} \tau_1 \goes{b_1} \tau_2 \goes{b_2} \cdots)) \, = \, \gamma_0\goes{c_0}\gamma_1\goes{c_1}\gamma_2\goes{c_2}\cdots\quad$ then \\
		  \hspace*{6pt} $\lambda_{X}({\tt ts}(x_0 \goes{a_0} x_1 \goes{a_1} x_2 \goes{a_2} \cdots)) \, = \, \gamma_0\goes{c_0}\gamma_1\goes{c_1}\gamma_2\goes{c_2}\cdots\quad$.
		  
Informally, $\lambda$ is defined compositionally with respect to $\S$-terms.
\item[(iv)] if $\lambda_X({\tt t}(x_0 \goes{a_0} x_1 \goes{a_1} x_2 \goes{a_2} \cdots)) \quad\!\! = \, \tau_0\goes{b_0}\tau_1\goes{b_1}\tau_2\goes{b_2}\cdots\quad$ then for every $i\in\mathbb{N}$, \\
$\lambda_X(\overline{\tau_i}) = \tau_i\goes{b_i}\tau_{i+1}\goes{b_{i+1}} \cdots$, where $\overline{\tau_i}\in\Ss\Bb X$ arises from $\tau_i\in\Ss X$ by replacing every $x_j$ with the stream starting at it. Informally, $\lambda_X$ is defined ``decompositionally'' with respect to streams.
\end{itemize}

\end{example}

\section{BiGSOS laws and mixed-GSOS specifications}

In~\cite{turiplotkin} it was proved that (1) every GSOS law induces a distributive law and (2) every coGSOS law induces a distributive law. The two ways of inducing distributive laws explained there are both natural and convincing, but formally different. We wish to study the problem of inducing distributive laws from specifications that would generalize both GSOS and coGSOS laws, so we need to have a general understanding of what it means to induce a distributive law. To this end, we consider the following simple generalization of GSOS and coGSOS:

\begin{definition}\label{def:bigsos}\rm
Given endofunctors $\S$ and $B$ such that the free monad $\Ss$ over $\S$ exists and the cofree comonad $\Bb$ over $B$ exist, a {\em biGSOS law} is a natural transformation $\rho:\S\Bb\To B\Ss$.
\end{definition}

GSOS and coGSOS laws give rise to biGSOS laws by composing with injections or projections:
\[
\xymatrix{
	\S\Bb\ar@{=>}[r]^{\S\langle{\epsilon,\pi}\rangle} & \S\Bid\ar@{=>}[r]^{\rho'} & B\Ss,
} \qquad
\xymatrix{
    \S\Bb\ar@{=>}[r]^{\rho''} & B\Sid\ar@{=>}[r]^-{B[\eta,\iota]} & B\Ss.
}
\]
where $\rho'$ is a GSOS law and $\rho''$ is a coGSOS law.
(Note that $\langle\epsilon,\pi\rangle:\Bb\To\Bid$ and $[\eta,\iota]:\Sid\To\Ss$ are natural transformations.)
As a result, biGSOS laws generalize both GSOS and coGSOS laws. However, they offer much more flexibility. In particular, for the case of stream systems and LTSs, we consider:

\begin{definition}\rm
A stream (or LTS) specification is {\em mixed-GSOS} if every rule in it is either a GSOS rule or a coGSOS rule, and moreover, for any operator ${\tt f}$, rules that define {\tt f} (i.e., those that have {\tt f} on the left side of the conclusion) are either all GSOS or all coGSOS.
\end{definition}

Note that we allow coGSOS-defined operations in conclusions of GSOS rules (and vice versa), so that e.g. the specification in Example~\ref{ex:ivbuqwvw} below is mixed-GSOS.

One could also define mixed GSOS more abstractly, by partitioning the signature into two disjoint subsignatures, $\Sigma=\Sigma_{\rm GSOS}+\Sigma_{\rm coGSOS}$, and requesting two natural transformations:
\[
	\rho_{\rm GSOS}:\Sigma_{\rm GSOS}\Bid\To B\Ss \qquad\qquad \rho_{\rm coGSOS}:\Sigma_{\rm coGSOS}\Bb\To B\Sid,
\]
one responsible for the GSOS, the other one the coGSOS part of the specification. It is then clear how a mixed-GSOS specification induces a biGSOS law, by comparing $\rho_{\rm GSOS}$ and $\rho_{\rm coGSOS}$ composed with suitable injections and projections. Note that biGSOS laws allow still more flexibility than allowed by mixed-GSOS, as they allow rules that combine complex conclusion terms as in GSOS, with lookahead as in coGSOS. 

It may not be evident what it means for a biGSOS law to induce a distributive law, but it is clear how a given distributive law may extend a biGSOS law, by composing with relevant injections and projections:
\begin{definition}\label{def:extn}\rm
A distributive law $\lambda:\Ss\Bb\To\Bb\Ss$ {\em extends} a biGSOS law $\rho:\S\Bb\To B\Ss$ if the following diagram commutes:
\begin{equation}\label{eq:ibvneverv}
\vcenter{\xymatrix{
	\Sigma\Bb\ar@{=>}[r]^{\rho}\ar@{=>}[d]_{\iota\Bb} & B\Ss \\
	\Ss\Bb\ar@{=>}[r]_{\lambda} & \Bb\Ss\ar@{=>}[u]_{\pi\Ss}
}}\end{equation}
\end{definition}

In other words, $\lambda$ extends $\rho$ if it equals $\rho$ when its arguments are restricted to $\Sigma$-terms of depth $1$ and results projected to $B$-behaviours of depth $1$.

As the following examples show, not every biGSOS law extends to a distributive law, and those that do may not extend uniquely.

\begin{example}\label{ex:ivbuqwvw}\rm
For $BX=A\times X$ with a chosen element $\$\in A$, consider syntax with one constant ${\tt C}$ and one unary operation ${\tt q}$, so that $\S X = 1+X$ and $\Bb X=(X\times A)^{\omega}$. Consider $\rho:\S\Bb\To B\Ss$ defined by rules:
\begin{align*}
	\dfrac{}{{\tt C} \goes{\labA} {\tt q(C)}} \qquad \dfrac{x\goes{a}x'\goes{b}x''}{{\tt q}(x)\goes{b}{\tt q}(x'')}\ (\mbox{for }a,b\in A)
\end{align*}
Consider any distributive law $\lambda:\Ss\Bb\To\Bb\Ss$, and present $\lambda_0({\tt C})$ as:
\begin{equation}\label{eq:oeirnvwqv}
	\lambda_0({\tt C}) \quad = \quad {\tt C}\goes{a_0} \tau_1\goes{a_1}\tau_2\goes{a_2}\cdots \quad \in \quad \Bb\Ss0
\end{equation}
with each $\tau_i\in\Ss0$ and $a_i\in A$. 

If $\lambda$ extends $\rho$ then, by~\eqref{eq:ibvneverv} applied to ${\tt C}\in\Sigma\Bb0$, we have $a_0=\labA$ and $\tau_1={\tt q(C)}$.
Since $\lambda$ is a distributive law, by axioms (ii) and (iv) of Definition~\ref{def:dist-law} as explained in Example~\ref{ex:dist-laws}, from~\eqref{eq:oeirnvwqv} we get
\begin{equation}\label{eq:vinvewv}
	\lambda_0({\tt q(C)}) \quad = \quad {\tt q(C)}\goes{a_1}\tau_2\goes{a_2}\cdots
\end{equation}
Now, by~\eqref{eq:ibvneverv} applied to ${\tt q}(\lambda_0({\tt C}))\in\Sigma\Bb\Ss0$, we have:
\begin{equation}\label{eq:oesvnaefv}
	\lambda_{\Ss0}({\tt q}(\lambda_0({\tt C}))) = \lambda_{\Ss0}({\tt q}({\tt C}\goes{a_0} \tau_1\goes{a_1}\tau_2\goes{a_2}\cdots)) = {\tt q(C)}\goes{a_1} {\tt q}(\tau_2)\goes{}\cdots
\end{equation} 
(only the first step of the stream on the right is determined this way). By axiom (iii) of Definition~\ref{def:dist-law} as explained in Example~\ref{ex:dist-laws}, the stream~\eqref{eq:oesvnaefv} is equal to~\eqref{eq:vinvewv} (or, more precisely, it is mapped to it by pointwise application of $\mu_0$); as a result, ${\tau_2}={\tt q}(\tau_2)$. However, there is no such term $\tau_2$ and, as a consequence, a distributive law $\lambda$ that extends $\rho$ does not exist.
\end{example}

\begin{example}\label{ex:onvawrar}\rm
Consider the previous example with the rightmost rule slightly modified to:
\begin{align*}
	\dfrac{x\goes{a}x'\goes{b}x''}{{\tt q}(x)\goes{b}x''}\ (\mbox{for }a,b\in A)
\end{align*}
If, say, $A=\{\labA,\labB\}$, then the corresponding $\rho$ can be extended e.g. to distributive laws $\lambda,\lambda'$ such that:
\begin{align*}
\lambda_0({\tt C}) &= {\tt C}\goes{\labA}{\tt q(C)}\goes{\labA}{\tt q(C)}\goes{\labA}{\tt q(C)}\goes{\labA}\cdots \\
\lambda'_0({\tt C}) &= {\tt C}\goes{\labA}{\tt q(C)}\goes{\labB}{\tt q(C)}\goes{\labB}{\tt q(C)}\goes{\labB}\cdots \\
\end{align*}
\end{example}

This example shows that distinct distributive laws $\lambda,\lambda':\Ss\Bb\To\Bb\Ss$ can sometimes be equalized by composing with both $\iota\Bb:\S\Bb\To\Ss\Bb$ and $\pi\Ss:\Bb\Ss\To B\Ss$ (see Definition~\ref{def:extn}). However, distinct distributive laws cannot be equalized by composing with only one of these transformations:

\begin{lemma}\label{lem:wonvaev}\rm
For any distributive laws $\lambda,\lambda':\Ss\Bb\To\Bb\Ss$:
\begin{center}
 (a)\ \ if $\lambda\circ\iota\Bb=\lambda'\circ\iota\Bb$ then $\lambda=\lambda'$, \qquad and \qquad (b)\ \  if $\pi\Ss\circ\lambda=\pi\Ss\circ\lambda'$ then $\lambda=\lambda'$.
\end{center}
\end{lemma}

It makes sense to say that a biGSOS law $\rho$ induces a distributive law if there is a unique distributive law that extends $\rho$. This is consistent with known results about GSOS and coGSOS laws, which, as has been understood since~\cite{turiplotkin}, induce distributive laws:

\begin{theorem}\label{thm:GSOSextends}\rm
For every GSOS law $\rho:\S\Bid\To B\Ss$, and for every coGSOS law $\rho:\S\Bb\To B\Sid$  there is a unique distributive law $\lambda:\Ss\Bb\To\Bb\Ss$ that extends the associated biGSOS law.
\end{theorem}

\begin{proof}[Proof sketch]
For the existence of $\lambda$, constructions of distributive laws from GSOS and coGSOS laws were given already in~\cite{turiplotkin}, and later explained more elegantly in~\cite{lenisapowerwatanabe2}. It is not difficult to prove that those constructions extend the respective GSOS and coGSOS laws in the sense of Definition~\ref{def:extn}.

For the uniqueness of $\lambda$, Lemma~\ref{lem:wonvaev} is used.
\end{proof}

\section{Queue machines}\label{sec:queue}

We shall prove that it is undecidable whether a given biGSOS law uniquely extends to a distributive law. To this end, we use the undecidability of the halting problem of queue machines.

A queue machine (QM) is a deterministic finite automaton additionally equipped with a first-in-first-out queue to store letters. A machine can read letters off the queue, and depending on their contents change their state while adding new letters to the queue.
Under the classical definition~\cite{kozen}, a QM in each transition (a) removes exactly one letter from the queue and (b) adds some (possibly zero) letters to it. 
For our purposes, it will be convenient to consider instead a variant of QMs that, in each step:
(a) remove zero, one or two letters from the queue, and
(b) add exactly one letter to it. 
Formally:
\begin{definition}\label{MQM_def}\rm
A {\em queue machine} (QM) ${\cal M}=(Q,A, \$, q_1, \delta_0, \delta_1, \delta_2)$ consists of a finite set $Q$ of states, a finite alphabet $A$ with a chosen symbol $\$\in A$, a starting state $q_1\in Q$, and three partial transition functions:
\[
	\delta_0: Q \pto Q\times A \qquad \delta_1: Q \times A \pto Q\times A \qquad \delta_2: Q \times A \times A \pto Q\times A
\]
that are disjointly defined and jointly total, i.e., such that for each $q\in Q$ and $a,b\in A$, exactly one of $\delta_0(q)$, $\delta_1(q,a)$ or $\delta_2(q,a,b)$ is defined.
A {\em configuration} of ${\cal M}$ is a pair $(q,w)\in Q\times A^*$; the machine induces a transition function $\mqmgoes{}$ on the set of configurations by:
\begin{align*}
	(q,w)  \mqmgoes{} (q',wc) \qquad &\mbox{if} \qquad \delta_0(q)=(q',c) \\
	(q,aw)  \mqmgoes{} (q',wc) \qquad &\mbox{if} \qquad \delta_0(q) \mbox{ undefined and }\delta_1(q,a)=(q',c) \\
	(q,abw) \mqmgoes{} (q',wc) \qquad &\mbox{if} \qquad \delta_0(q) \mbox{ and } \delta_1(q,a) \mbox{ undefined and }\delta_2(q,a,b)=(q',c).
\end{align*}
\end{definition}
Note that an MQM never makes a queue empty, and it terminates if and only if it reaches a configuration $(q,a)$ with a single letter $a$ in the queue, such that $\delta_0(q)$ and $\delta_1(q,a)$ are undefined.

\begin{theorem}\label{thm:qmundecid}\rm
It is undecidable whether a given QM terminates from the configuration $(q_1,\$)$, called the {\em initial configuration}.
\end{theorem}
\begin{proof}
As is well known, it is undecidable whether a classical QM ${\cal M}$ as defined in~\cite{kozen} terminates on its initial configuration.
For every classical ${\cal M}$ one constructs a QM $\overline{\cal M}$ as in Definition~\ref{MQM_def} that terminates on its initial configuration if and only if ${\cal M}$ does.
\end{proof}

\section{From queue machines to stream specifications}\label{sec:qm2ss}\rm

Given a QM ${\cal M}=(Q,A, \$, q_1, \delta_0, \delta_1, \delta_2)$, consider a signature with a single constant ${\tt C}$ and a family of unary operation symbols $\{{\tt q}\mid q\in Q\}$, and a family of rules:
\begin{equation}\label{eq:ssrules}
\dfrac{}{{\tt C}\goes{\$}{\tt q_1(C)}}\ (\textbf{C})
\qquad
\dfrac{}{{\tt q}(x)\goes{c}{\tt q'}(x)}\ (\textbf{R0})
\qquad
\dfrac{x\goes{a}y}{{\tt q}(x)\goes{c}{\tt q'}(y)}\ (\textbf{R1})
\qquad
\dfrac{x\goes{a}y\goes{b}z}{{\tt q}(x)\goes{c}{\tt q'}(z)}\ (\textbf{R2})
\end{equation}
for all $q,q'\in Q$ and $a,b,c\in A$ subject to the following conditions:
\begin{itemize}
\item {\bf R0} is included when $\delta_0(q)=(q',c)$,
\item {\bf R1} is included when $\delta_0(q)$ is undefined and $\delta_1(q,a)=(q',c)$, and
\item {\bf R2} is included when $\delta_0(q)$ and $\delta_1(q,a)$ are undefined and $\delta_2(q,a,b)=(q',c)$.
\end{itemize}
These rules are mixed GSOS, so they define a biGSOS law $\rho_{\cal M}:\S\Bb\To B\Ss$, where $BX=A\times X$ and $\S X = 1+Q\times X$.
%
%
We shall now prove, in a sequence of lemmas, that $\rho_{\cal M}$ uniquely extends to a distributive law if and only if ${\cal M}$ does {\em not} terminate from the initial configuration. Our argument relies on the following correspondence between partial runs of ${\cal M}$ and prefixes of streams produced by distributive laws that extend $\rho_{\cal M}$:

\begin{lemma}\label{lem:owefnw}\rm
For every $n>0$, if a QM ${\cal M}$ makes $n-1$ steps from the initial configuration:
\[
	q_1,w_1 \mqmgoes{} q_2,w_2 \mqmgoes{} q_3,w_3 \mqmgoes{}\cdots\mqmgoes{}q_n,w_n
\]
(where $w_1=\$$) then every distributive law $\lambda$ that extends $\rho_{\cal M}$ maps the constant symbol ${\tt C}\in \Ss\Bb0$ to a stream $\lambda_0({\tt C})\in\Bb\Ss0$ that begins with
\[
	\tau_0\goes{\$}\tau_1\goes{a_1}\tau_2\goes{a_3}\tau_3\goes{a_3}\cdots\goes{a_{n-1}}\tau_n,
\]
where 
\begin{itemize}
\item each $a_i\in A$ is the last letter of $w_{i+1}$, i.e., the letter added to the queue in the $i$-th step of ${\cal M}$,
\item $\tau_0={\tt C}$, and $\tau_1,\ldots,\tau_n\in\Ss0$ are such that each $\tau_i={\tt q}_i(\tau_j)$, where $0\leq j<i$ is such that $i-j=|w_i|$.
\end{itemize}
(Note that from these properties it follows that $a_ja_{j+1}\cdots{a_{i-1}}=w_i$.)
\end{lemma}

\begin{proof}
We proceed by induction on $n$. For the base case $n=1$, if $\lambda$ extends $\rho_{\cal M}$ then, thanks to rule {\bf C}, the stream $\lambda_0({\tt C})$ must begin with:
\[
\lambda_0({\tt C})={\tt C}\goes{\$}{\tt q}_1({\tt C})
\]
which satisfies the inductive statement.

For the inductive step, assume that ${\cal M}$ makes $n$ steps:
\[
	q_1,w_1 \mqmgoes{} q_2,w_2 \mqmgoes{} q_3,w_3 \mqmgoes{}\cdots\mqmgoes{}q_n,w_n\mqmgoes{}q_{n+1},w_{n+1}
\]
By the inductive assumption, for any $\lambda$ that extends $\rho_{\cal M}$, the stream $\lambda_0({\tt C})$ must begin with:
\begin{equation}\label{eq:taus}
	\tau_0\goes{\$}\tau_1\goes{a_1}\tau_2\goes{a_3}\tau_3\goes{a_3}\cdots\goes{a_{n-1}}\tau_n,
\end{equation}
where $\tau_n={\tt q}_n(\tau_j)$ such that $n-j=|w_n|$, and $w_n=a_ja_{j+1}a_{j+2}\cdots{a_{n-1}}$.

There are three cases to consider, depending on how the configuration $(q_{n+1},w_{n+1})$ is derived from $(q_n,w_n)$:
\begin{itemize}
\item $\delta_0(q_n)=(q_{n+1},a_n)$ and $w_{n+1}=w_na_n$, for some $a_n\in A$. Then $\rho_{\cal M}$ includes a corresponding rule {\bf R0}, and if $\lambda$ extends $\rho_{\cal M}$ then the initial part~\eqref{eq:taus} in $\lambda_0({\tt C})$ is necessarily extended with $\tau_n\goes{a_n}\tau_{n+1}={\tt q}_{n+1}(\tau_{j})$, and the inductive statement is preserved.
\item $\delta_0(q_n)$ is undefined, and $\delta_1(q_n,a_j)=(q_{n+1},a_n)$ and $w_{n+1}=a_{j+1}a_{j+2}\cdots a_{n-1}a_n$, for some $a_n\in A$. Then $\rho_{\cal M}$ includes a corresponding rule {\bf R1}, and if $\lambda$ extends $\rho_{\cal M}$ then the initial part~\eqref{eq:taus} in $\lambda_0({\tt C})$ is necessarily extended with $\tau_n\goes{a_n}\tau_{n+1}={\tt q}_{n+1}(\tau_{j+1})$, and the inductive statement is preserved.
\item $\delta_0(q_n)$ and $\delta_1(q_n,a_j)$ are undefined, and $\delta_2(q_n,a_j,a_{j+1})=(q_{n+1},a_n)$ and $w_{n+1}=a_{j+2}\cdots a_{n-1}a_n$, for some $a_n\in A$. (Note that, since {\cal M} does not terminate in $(q_n,w_n)$, we know that $n-j\geq 2$.) Then $\rho_{\cal M}$ includes a corresponding rule {\bf R2}, and if $\lambda$ extends $\rho_{\cal M}$ then the initial part~\eqref{eq:taus} in $\lambda_0({\tt C})$ is necessarily extended with $\tau_n\goes{a_n}\tau_{n+1}={\tt q}_{n+1}(\tau_{j+2})$, and the inductive statement is preserved.
\end{itemize}
\end{proof}

\begin{lemma}\label{lem:ibvwavv}\rm
For any QM ${\cal M}$ that does not terminate from the initial configuration, the transformation $\rho_{\cal M}$ is extended by at most one distributive law.
\end{lemma}
\begin{proof}
Consider distributive laws $\lambda,\lambda':\Ss\Bb\To\Bb\Ss$ that both extend $\rho_{\cal M}$. For any set $X$, we wish to prove that the component functions $\lambda_X,\lambda'_X:\Ss\Bb X\to \Bb\Ss X$ are equal. We prove this by structural induction on terms $t\in\Ss\Bb X$.

For the first base case, if $t=\sigma\in\Bb X$ then $\lambda_X(t)=\lambda'_X(t)$ follows immediately from axiom (i) of Definition~\ref{def:dist-law}.
For the second base case, if $t={\tt C}$ then the equality follows from Lemma~\ref{lem:owefnw}, since ${\cal M}$ makes arbitrarily many steps from the initial configuration.

For the inductive step, we need to prove that for all terms $t\in\Ss\Bb X$ and states $q\in Q$, if $\lambda_X(t)=\lambda'_X(t)$ then $\lambda_X({\tt q}(t))=\lambda'_X({\tt q}(t))$. Denote 
\[
	\sigma \quad =\quad  \lambda_X(t)\quad =\quad \lambda'_X(t) \quad = \quad \tau_0\goes{a_0}\tau_1\goes{a_1}\tau_2\goes{a_2}\cdots \qquad (\tau_i\in\Ss X).
\]
We begin by proving that the desired equality holds when postcomposed with $\pi_{\Ss X}:\Bb\Ss X\to B\Ss X$, i.e., that the streams $\lambda_X({\tt q}(t))$ and $\lambda'_X({\tt q}(t))$ coincide on their first transitions. This is proved by case analysis similar to that used in the proof of Lemma~\ref{lem:owefnw}. For example, if $\delta_0(q)$ is undefined and $\delta_1(q,a_0)=(q',b)$ for some $q'\in Q$ and $b\in A$, then $\rho_{\cal M}$ includes a relevant {\bf R1} rule and if $\lambda$ and $\lambda'$ both extend $\rho_{\cal M}$ then $\lambda_X({\tt q}(t))$ and $\lambda'_X({\tt q}(t))$ must both begin with ${\tt q}(\tau_0)\goes{b}{\tt q'}(\tau_1)$.

We proved that for any term $t \in \Ss \Bb X$ the streams $\lambda_X(t)$ and $\lambda'_X(t)$ coincide on the first transitions, i.e., $\pi\Ss\circ\lambda=\pi\Ss\circ\lambda'$. Hence, by Lemma \ref{lem:wonvaev}(b), $\lambda=\lambda'$.
\end{proof}

\begin{lemma}\label{lem:awonbaerb}\rm
If a QM ${\cal M}$ does not terminate from the initial configuration, then there exists a distributive law that extends $\rho_{\cal M}$.
\end{lemma}
\begin{proof}
Fix a QM ${\cal M}$ that does not terminate from the initial configuration $(q_1,\$)$. We shall define a distributive law $\lambda$ that extends $\rho_{\cal M}$. For any set $X$, begin by defining
\[
	\lambda_X({\tt C})  \quad = \quad \tau_0 \goes{a_0} \tau_1\goes{a_1}\tau_2\goes{a_2}\cdots \quad \in \quad \Bb \Ss X
\]
with $\tau_i\in\Ss X$ and $a_i\in A$ such that:
\begin{itemize}
\item $\tau_0={\tt C}$ and $a_0=\$$, 
\item for any $i>0$, $\tau_i={\tt q}_i(\tau_j)$, where the $i$-th configuration reached by ${\cal M}$ is $(q_i,w_i)$ and $j=i-|w_i|$; moreover, $a_j$ is the first letter of $w_i$.
\end{itemize} 
To define $\lambda_X$ on other terms in $\Ss \Bb X$, note that apart from rule {\bf C}, the entire specification $\rho_{\cal M}$ is a coGSOS specification, therefore, by Theorem~\ref{thm:GSOSextends}, there exists a distributive law $\hat\lambda$ that extends all rules of $\rho_{\cal M}$ apart from {\bf C}. For any term ${\tt t}\in\Ss\Bb X$ where ${\tt C}$ does not appear, define $\lambda_X({\tt t})$ to be $\hat\lambda_X({\tt t})$. If ${\tt C}$ appears in ${\tt t}$, replace it with the stream $\lambda_X({\tt C})$ and use $\hat\lambda_{\Ss X}$ followed by $\Bb\mu_X$ on the term obtained.

It is easy to see that $\lambda$ defined in this manner is natural and satisfies axioms (i)-(iii) of Definition~\ref{def:dist-law} (see also Example~\ref{ex:dist-laws}). 

The only remaining axiom is (iv), which in principle could fail if the above procedure, on one of the terms $\tau_i$ present in $\lambda_X({\tt C})$, returned a stream that differs from the substream of $\lambda_X({\tt C})$ starting at $\tau_i$. This is, however, not the case, as can be proved by induction on $i$, using case analysis similar to that used in the proof of Lemma~\ref{lem:owefnw}.
\end{proof}

\begin{lemma}\label{lem:wbvvw}\rm
If a QM ${\cal M}$ terminates from the initial configuration, then there is no distributive law that extends $\rho_{\cal M}$.
\end{lemma}
\begin{proof}
Assume to the contrary, that $\cal M$ terminates after $n$ steps in a configuration $(q_n,w_n)$ and there is a distributive law $\lambda$ that extends $\rho_{\cal M}$. By Lemma \ref{lem:owefnw}, the stream $\lambda_0({\tt C})$ begins with:
\[
{\tt C} \goes{a_1}\tau_1\goes{a_2}\tau_2\goes{a_3}\cdots\tau_{n-1}\goes{a_n}\tau_n
\] 
where $\tau_n={\tt q}_n(\tau_{n-|w_n|})$. Note that $\cal M$ can terminate in $(q_n,w_n)$ only if $w_n$ has length 1, hence $\tau_n={\tt q}_n(\tau_{n-1})$ and $w_n=a_n$; moreover, $\delta_0(q_n)$ and $\delta_1(q_n,a_n)$ must be undefined.

The remaining argument follows the line of Example~\ref{ex:ivbuqwvw}. Suppose that the next step in $\lambda_0({\tt C})$ is $\tau_n\goes{a_{n+1}}\tau_{n+1}$, for some $a_{n+1}\in A$ and $\tau_{n+1}\in\Ss 0$. Since $\delta_0(q_n)$ and $\delta_1(q_n,a_n)$ are undefined, $\delta_2(q_n,a_n,a_{n+1})=(q',b)$ must be defined for some $q'\in Q$ and $b\in A$. As a result, $\rho_{\cal M}$ contains an {\bf R2} rule:
\[
\dfrac{x\goes{a_n}y\goes{a_{n+1}}z} {{\tt q}_n(x)\goes{b}{\tt q'}(z)}
\]
and, since $\lambda$ extends $\rho_{\cal M}$, instantiating $x$ to $\tau_{n-1}$ we obtain $b=a_{n+1}$ and $\tau_{n+1}={\tt q'}(\tau_{n+1})$, a contradiction.
%
%
\end{proof}

Note that all rules in $\rho_{\cal M}$ are either GSOS or coGSOS rules; we call specifications with this property {\em mixed-GSOS specifications}. We arrive at a proof of our Claim from the Introduction:

\begin{theorem}\label{thm:ssundec}\rm
For the case of stream systems,
it is undecidable whether a given mixed-GSOS specification extends to a unique distributive law.
\end{theorem}
\begin{proof}
Combine Lemmas~\ref{lem:ibvwavv}-\ref{lem:wbvvw} with Theorem~\ref{thm:qmundecid}
\end{proof}

\section{Labelled transition systems}\label{sec:qm2lts}

We shall now show how to encode Queue Machines into mixed-GSOS specifications for LTSs, to prove that distributive laws admit no format for $BX=\Pf(A\times X)$ either. Since the general idea and most technical details are the same as in the case of stream systems (Section~\ref{sec:qm2ss}), we only sketch the differences between the two cases.

To begin with, note that the set of rules~\eqref{eq:ssrules} from Section~\ref{sec:qm2ss} can be read as rules in the mixed-GSOS format for $BX=\Pf(A\times X)$. However, taking the same rules for a QM ${\cal M}$ would give rise to a biGSOS law that always extends to some distributive law (a counterpart of Lemma~\ref{lem:wbvvw} would fail). Intuitively, unlike in the case of $BX=A\times X$, a distributive law for $BX=\Pf(A\times X)$ is allowed to produce an empty set of successors for a term that corresponds to a terminating configuration of ${\cal M}$.

Our solution is to extend the specification~\eqref{eq:ssrules}, now understood as a mixed-GSOS specification for the LTS behaviour, with additional rules:
\begin{equation}\label{eq:negrule}
\dfrac{x\goes{a}y\quad y{\not \goes{} }}{{\tt q}(x)\goes{a}{\tt q}(x)}\ (\textbf{R2'})
\end{equation}
for $q\in Q$ and $a\in A$. These new rules are included whenever $\delta_0(q)$ and $\delta_1(q,a)$ are undefined.
We denote the biGSOS law defined by the extended specification by $\rho_{EXT}:\S\Bb\To B\Ss$, where $BX=\Pf(A\times X)$ and $\S X = 1+Q\times X$.

For any QM $\cal M$, the biGSOS law $\rho_{EXT}$ uniquely extends to a distributive law if and only if $\cal M$ does not terminate from the initial configuration. The proof of this follows the line of Section~\ref{sec:qm2ss}, and we shall only explain the main differences here.

The main technical step in Section~\ref{sec:qm2ss}, Lemma~\ref{lem:owefnw}, holds in a very similar form:

\begin{lemma}\label{lem:osnvvsv}\rm
For every $n>0$, if a QM ${\cal M}$ makes $n-1$ steps from the initial configuration as in Lemma~\ref{lem:owefnw}, 
then every distributive law $\lambda$ that extends $\rho_{EXT}$ maps the constant symbol ${\tt C}\in \Ss\Bb0$ to a tree $\lambda_0({\tt C})\in\Bb\Ss0$ that begins with
a degenerate tree, i.e., a sequence:
\[
	\tau_0\goes{\$}\tau_1\goes{a_1}\tau_2\goes{a_3}\tau_3\goes{a_3}\cdots\goes{a_{n-1}}\tau_n
\]
where $a_i$ and $\tau_i$ are as in Lemma~\ref{lem:owefnw}.
\end{lemma}
\begin{proof}
By induction on $n$ entirely analogous to the proof of Lemma~\ref{lem:owefnw}. Intuitively, the initial part of $\lambda_0({\tt C})$ is degenerate because the specification $\rho_{EXT}$ is deterministic, i.e., it only infers one transition from ${\tt C}$, and infers at most one transition for ${\tt q}(x)$ if $x$ can make at most one transition.
\end{proof}

The next two lemmas are proved entirely analogously to Section~\ref{sec:qm2ss}:

\begin{lemma}\label{lem:wainbvwavaw}\rm
For an QM ${\cal M}$ that does not terminate from the initial configuration, the transformation $\rho_{EXT}$ is extended by at most one distributive law.
\end{lemma}

\begin{lemma}\label{lem:tyjdrtvx}\rm
If a QM ${\cal M}$ does not terminate from the initial configuration, then there exists a distributive law that extends $\rho_{EXT}$.
\end{lemma}

In particular, the distributive law defined in Lemma~\ref{lem:tyjdrtvx} is exactly as in the proof of Lemma~\ref{lem:awonbaerb}, with the streams produced in the latter considered as (degenerate) trees.

The only step that requires some care is Lemma~\ref{lem:wbvvw}, which now takes the form:

\begin{lemma}\label{lem:bnobbes}\rm
If a QM ${\cal M}$ terminates from the initial configuration, then there is no distributive law that extends $\rho_{EXT}$.
\end{lemma}
\begin{proof}
Assume to the contrary, that $\cal M$ terminates after $n$ steps in a configuration and there is a distributive law $\lambda$ that extends $\rho_{\cal M}$. By Lemma~\ref{lem:osnvvsv}, the tree $\lambda_0({\tt C})$ begins with a sequence:
\[
{\tt C} \goes{a_1}\tau_1\goes{a_2}\tau_2\goes{a_3}\cdots\tau_{n-1}\goes{a_n}\tau_n
\] 
where, as in the proof of Lemma~\ref{lem:wbvvw}, $\tau_n={\tt q}_n(\tau_{n-1})$, and $\delta_0(q_n)$ and $\delta_1(q_n,a_n)$ are undefined.

What successors can $\tau_n$ have in the tree $\lambda_0({\tt C})$? Assume first that is has no successors. Since $\lambda$ extends $\rho_{EXT}$, by applying a corresponding rule {\bf R2'} instantiated to ${\tt q}={\tt q}_n$, $x=\tau_{n-1}$ and $a=a_{n}$ we infer that $\tau_n={\tt q}(\tau_{n-1})$ indeed does have at least one successor, which is a contradiction.

Now assume that $\tau_n$ has some successors. All these successors are terms in $\Ss 0$. Some of these successors are minimal, i.e., have the smallest depth of nesting of operations $\tt q_i$. Pick one of these minimal successors and call it $\tau'$. Since $\lambda$ extends $\rho_{EXT}$, the transition $\tau_n\goes{b}\tau'$ must be derivable from rules in $\rho_{EXT}$. The only rule that can be used to this end is a corresponding rule {\bf R2}, instantiated to ${\tt q}={\tt q}_n$, $x=\tau_{n-1}$, $y=\tau_n$ and $a=a_{n}$. But this means that $\tau_n$ must have a successor $z$ such that $\tau'={\tt q'}(z)$, which contradicts the minimality of $\tau'$.
\end{proof}

Thus we prove our Claim from the Introduction for the case of LTSs:

\begin{theorem}\label{thm:ltsundec}\rm
For the case of labeled transition systems ($BX=\Pf(A\times X)$),
it is undecidable whether a given mixed-GSOS specification extends to a unique distributive law.
\end{theorem}
\begin{proof}
Combine Lemmas~\ref{lem:wainbvwavaw}-\ref{lem:bnobbes} and Theorem~\ref{thm:qmundecid}.
\end{proof}

\section{Related work}\label{sec:related}

We have proved, for the case of stream systems and LTSs, that there is no format for distributive laws of monads over comonads that would be complete for mixed-GSOS specification, i.e., that would cover exactly those mixed-GSOS specification that extend to a distributive law. The specifications used in our proofs are actually coGSOS specifications extended with only one GSOS rule that has no premises. Moreover, the coGSOS rules only uses lookahead of depth 2, and the GSOS rule uses a rule conclusion of height 2. As a result, there is no complete format even for such restricted specifications.

On the other hand, our results do not contradict the existence of formats complete for classes of specifications that do not cover the mixed-GSOS format. Indeed as shown in~\cite{statontyft}, in the context of LTSs one can combine GSOS and coGSOS but restrict to specifications with positive premises only, and guarantee the existence of a corresponding distributive law. (Note that specifications used in Section~\ref{sec:qm2lts} rely on negative rule premises.)

Our proofs can be easily modified to show undecidability of other problems related to operational specifications, some of them phrased without reference to distributive laws. For example, in the case of LTSs, it is undecidable whether a transition system specification (or even a mixed-GSOS specification) has a supported model, a unique supported model, or a unique stable model~\cite{vGnegative}; the constructions needed for these are minor variations of the one used in Section~\ref{sec:qm2lts}.

In the case of stream systems, our results are related to studies of the productivity of stream definitions~\cite{endrullis-pebble}. Specifications used in Section~\ref{sec:qm2ss} can be seen as definitions in the ``pure stream specification format'' of~\cite{endrullis-pebble}. Indeed, that format is closely related to stream coGSOS extended with premise-less GSOS rules for constants. In~\cite{endrullis-pebble} it was proved that productivity of pure stream specifications is decidable for specifications that are data-oblivious, i.e., natural with respect to transition labels. Our specifications are not data-oblivious in that sense. It is easy to use the constructions of Section~\ref{sec:qm2ss} to prove that productivity of pure stream specifications becomes undecidable without data-obliviousness.

\noindent
{\bf Acknowledgment.} We are grateful to Jurriaan Rot for several helpful discussions, and to anonymous referees for spotting embarrassing mistakes both in the content and the presentation of our results.
\bibliographystyle{eptcs}
\bibliography{noformat}

\end{document}